\documentclass[a4paper,conference]{IEEEtran}
\usepackage{amsmath,amsthm,amsfonts,amssymb,latexsym,extarrows,booktabs,array}
\usepackage{fancyhdr,graphicx,tabularx}
\usepackage{multirow}
\usepackage{lipsum}
\usepackage[top=0.75in, bottom=0.55in,left=0.6in, right=0.6in]{geometry}
\usepackage{amsfonts}
\usepackage{graphicx}
\usepackage{subfigure}
\usepackage {amsmath}
\usepackage {amssymb}
\usepackage {latexsym}
\usepackage {graphicx}
\usepackage {cite}
\usepackage {subfigure}
\usepackage {url}
\usepackage {stfloats}
\usepackage {mathrsfs}
\usepackage{algorithm}
\usepackage{algorithmic}
\usepackage{xcolor}

\usepackage{setspace}
\usepackage{listings}
\IEEEoverridecommandlockouts

\newcommand{\ls}[1]
    {\dimen0=\fontdimen6\the\font
     \lineskip=#1\dimen0
     \advance\lineskip.5\fontdimen5\the\font
     \advance\lineskip-\dimen0
     \lineskiplimit=.9\lineskip
     \baselineskip=\lineskip
     \advance\baselineskip\dimen0
     \normallineskip\lineskip
     \normallineskiplimit\lineskiplimit
     \normalbaselineskip\baselineskip
     \ignorespaces
}
\begin{document}

 \title{{ Optimal Grassmann Manifold  Eavesdropping: A Huge Security Disaster for  $M$-1-2\\ Wiretap Channels}}
\author{
\IEEEauthorblockN{Dongyang Xu\IEEEauthorrefmark{1}\IEEEauthorrefmark{2}, Pinyi Ren\IEEEauthorrefmark{1}\IEEEauthorrefmark{2}, and James A. Ritcey\IEEEauthorrefmark{3}}\\\vspace{-10pt}
 \IEEEauthorblockA{\ls{1.3}\IEEEauthorrefmark{1}School of Electronic and Information Engineering, Xi'an Jiaotong University, China} \\\vspace{-10pt}
\IEEEauthorblockA{\IEEEauthorrefmark{2}Shaanxi Smart Networks and Ubiquitous Access Research Center, China. }\\\vspace{-10pt}
\IEEEauthorblockA{\IEEEauthorrefmark{3}Department of Electrical Engineering, University of Washington, USA.}\\\vspace{-10pt}
\IEEEauthorblockA{E-mail: \{\emph {xudongyang@stu.xjtu.edu.cn, pyren@mail.xjtu.edu.cn, ritcey@ee.washington.edu}
\}}
}
\maketitle
\begin{abstract}
We in this paper introduce an advanced eavesdropper that aims to paralyze the artificial-noise-aided secure communications. We consider  the $M$-1-2 Gaussian MISO wiretap channel, which consists of a $M$-antenna transmitter, a single-antenna receiver,  and a two-antenna eavesdropper. This type of  eavesdropper, by adopting an optimal Grassmann manifold (OGM) filtering structure,  can  reduce the maximum achievable secrecy rate (MASR) to be zero  by using only two receive antennas, regardless of  the number of antennas  at the transmitter. Specifically, the eavesdropper  exploits  linear  filters to serially recover the legitimate  information symbols and intends to find the optimal  filter that  minimizes the mean-square error (MSE) in estimating the symbols. During the process, a convex semidefinite programming (SDP) problem with constraints on the filter matrix can be formulated  and  solved. Interestingly, the resulted optimal filters constitute a complex Grassmann manifold on the matrix space. Based on the filters, a novel expression of MASR is  derived and further verified to be zero  under the  noiseless environment. Besides this, an achievable variable region (AVR) that induces zero MASR is presented analytically in the noisy case. Numerical results are provided  to illustrate the huge disaster in the respect  of secrecy rate.

\end{abstract}
\vspace{5pt}
\section{Introduction}
\label{introduction}
The research on physical layer security has been drawing  great attentions with   various  wireless technologies  moving forward. This protection technique,   by exploiting the wireless channel variations, can enhance the information secrecy of  network  which is conventionally  protected by the cryptography in the upper layers of  network protocol stack. The initial research on the physical layer security  refers to the area of  information-theoretic secrecy in the wiretap channel which  consists of a source, a destination, and an eavesdropper that attempts to intercept the communication between the source and the destination~\cite{Wyner_AD}. The authors in~\cite{Leung-Yan-Cheong} investigated a  fundamental secrecy  measure called  secrecy capacity which is defined as  the maximum achievable secrecy rate (MASR) between the source and the destination while ensuring that no information can be inferred by the eavesdropper. It was shown that a non-zero secrecy capacity can be achieved between the source and the destination if the channel to the destination is better than that to the eavesdropper. Otherwise, the temporal variations of the channel and the additional degrees of multi-antennas   have  to be exploited  to maintain a positive secrecy rate~\cite{Yingbin_Liang,Amitav_Mukherjee}. The key point is that the perfect knowledge of  channel state information (CSI) of the main and eavesdropper channel  is required at the source, which however is  an unpractical assumption~\cite{Khisti_I}.

Therefore,  large amounts of research  have been  concentrated on the secure transmission  under  imperfect  acquisition of  eavesdropper's  CSI at the source.  Among  those work,  artificial noise (AN), being emitted on top of the information-bearing signal to disrupt the reception at the eavesdropper,  gradually comes to be the commonsense  of    guaranteeing high secrecy rate~\cite{Negi}.  With the aid of  multi-antenna channels, AN can be placed in the dimensions that cause least interference at the destination but engender dramatic difference between the signal qualities at the destination and the eavesdropper.   There exist huge number of AN-based transmission strategies, such as the precoder design and power allocation mechanisms~\cite{Zhou,Khisti_II,Nan_Yang,Jun_Zhu}. In a multiple-input-single-output-multiple-antenna eavesdropper (MISOME) system where eavesdroppers are equipped with less antennas than the source, authors in~\cite{Zhou} gave an analytical closed-form expression of an achievable secrecy rate and determined an optimal transmit power allocation between the legitimate signal and the AN to maximize the secrecy rate. Authors in~\cite{Khisti_II} studied the AN based scheme in MIMO wiretap channel scenario and presented that the achievable secrecy rate can however be arbitrarily far from the secrecy capacity if the eavesdropper has more antennas than the source and destination. Authors in~\cite{Nan_Yang} investigated the optimal physical-layer secure transmission with artificial noise in the wiretap channel and analyzed the impact of the number of transceiver antennas on the secrecy rate. Authors in~\cite{Jun_Zhu} considered secure downlink transmission in an AN-aided multicell massive MIMO system with a passive multi-antenna eavesdropper. It shows that even with AN assistance, secure transmission can not be guaranteed if the number of eavesdropper antennas is too large. The above AN based schemes assumed that the eavesdropper usually adopt a theoretically capacity-achievable receiver to combat AN.  Note that this type of receiver is  originally designed without considering the artificial interference. However, there are few literatures discussing the effect of  receiver design on the AN based schemes, especially from an eavesdropping  perspective.   The resulted eavesdropper is therefore underestimated  and  mistaken
to have to eliminate AN interference completely  by consuming large antenna resources.

To give a more comprehensive  understanding of the functionality of AN  mechanism, we in this paper first introduce an optimal Grassmann manifold (OGM) receiver  at the eavesdropper  which can exploit the  statistical property of received data to combat AN.  We can show that secrecy protection usually achieved by AN  can be completely broken down  by the eavesdropper with only two receive antennas, no matter how many antennas have been equipped at the source.  Our contributions are detailed as follows:
\begin{enumerate}
\item Based on the covariance of  received data,  a whitening operation is performed on the received data to eliminate the power influence of AN. Thereafter,  a vector-based  projector is derived so that  the eavesdropper can identify the phase information of secrecy information symbols from the whitened data. Finally,  an explicit expression of recovered data is  presented but subjected to an unknown filtering matrix.
\item Based on  the calculation of  the mean-squared error (MSE) of estimated symbols, an optimization problem of minimizing the MSE  can be formulated in which the filter matrix  serves as  the optimized variable. Particularly, this problem can be  cast as a convex  semidefinite program  (SDP) and solved with  a closed-form optimal solution. Interestingly, the solution includes a series of filters which  are verified  to lie on a complex Grassmann manifold.
\item Based on the optimal filter structure, the novel expression of MASR is derived. In the noiseless case,  the optimal power strategy of achieving MASR  is transformed into  no power allocation for AN and  zero secrecy rate can be always induced. This security disaster is  further confirmed  through simulations in a wide range of node locations and under a large number of transmit antennas. In the noisy case, an expression of  achievable variable region (AVR) is derived and essentially constituted by several factors, such as, a controllable unitary matrix,  the eavesdropper's noise level and  its distance to the source. To be worse, the eavesdropper can control AVR and then cause  zero secrecy rate no matter how many antennas are utilized at the source. We also simulate the  AVR  and confirm its existence even when the legitimate destination is located in various regions.
\end{enumerate}
The remainder of this paper is structured as follows. Section~\ref{System Description} presents the system model overview. A novel eavesdropping  threat and its technical modeling is given in Section~\ref{NETTM}. An optimal grassmann eavesdropping  manifold is constructed in  Section~\ref{COGEM}.  Secrecy performance analysis is performed in Section~\ref{SPAADE}. Simulation results are presented in Section~\ref{Simulation Results}. Finally, we conclude our work in Section~\ref{Conclusions}.

Notation: Boldface is used for vector ${\bf{a}}$ and matrix ${\bf{A}}$.  ${{\bf{A}}^{{H}}}$ respectively denotes   the conjugate transpose of matrix ${\bf{A}}$. ${\bf{A}}\left( {:,n} \right)$  represents the $n$-th column of ${\bf{A}}$. $\left\| {\cdot} \right\|$ denotes the Euclidean norm of a vector or a matrix. ${\mathbb{E}}\left\{  \cdot  \right\}$ is the expectation operator. ${\rm{Re}}\left\{  \cdot  \right\}$ and $\left| {\cdot} \right|$  respectively represent  the real part and  the modulus of a complex number.
\section{System Description}
\label{System Description}
 We consider a $M$-1-2 wiretap channel  where a transmitter Alice  with $M$ antennas  communicates with  a single-antenna legitimate receiver Bob and  is also eavesdropped  by an eavesdropper Eve with two antennas satisfying $ 2<M$. We denote the $i$-th transmitted symbol at Alice as  ${{\bf{s}}_i} \in {C^{M \times 1}},i = 1, \ldots, T$.  $T$ represents the total number of symbols (or channel uses) within each coherence interval. The symbols received at Bob  and Eve for $i = 1, \ldots, T$ are respectively  given by:
\begin{equation}
\begin{array}{l}
{{\bf{y}}_{{\rm{B}},i}} = {{\bf{h}}_{{\rm{B}}}}{{\bf{s}}_i} + {{\bf{n}}_{{\rm{B}},i}}\\
{{\bf{y}}_{{\rm{E}},i}} = {{\bf{H}}_{{\rm{E}}}}{{\bf{s}}_i} + {{\bf{n}}_{{\rm{E}},i}}
\end{array}
\end{equation}
where ${{\bf{h}}_{{\rm{B}}}} \in {{\mathbb C}^{1 \times {M}}}$ denotes the channel from Alice to Bob. ${{\bf{H}}_{{\rm{E}}}} \in {{\mathbb C}^{2 \times {M}}}$ denotes  the channel from Alice to Eve. The  elements of ${{\bf{h}}_{{\rm{B}}}}$ and ${{\bf{H}}_{{\rm{E}}}}$ are assumed to be  independent zero-mean complex Gaussian random  variables with unit variance.   The components of ${{\bf{n}}_{{\rm{B}},i}}$ and ${{\bf{n}}_{{\rm{E}},i}}$  are the equivalent noise with power $\sigma _{\rm{B}}^2$ and $\sigma _{\rm{E}}^2$, respectively. We assume  ${{\bf{n}}_{{\rm{B}},i}} = \beta _{\rm{B}}^{{ - 1/2}}{\widetilde {\bf{n}}_{{\rm{B}},i}}$ and ${{\bf{n}}_{{\rm{E}},i}} = \beta _{\rm{E}}^{ { - 1/2}}{\widetilde {\bf{n}}_{{\rm{E}},i}}$ where ${\beta _{\rm{B}}}$ and ${\beta _{\rm{E}}}$ are   path loss  factors respectively from Alice  to Bob and Eve. ${\widetilde {\bf{n}}_{{\rm{B}},i}}$ and ${\widetilde {\bf{n}}_{{\rm{E}},i}}$ are i.i.d. additive white Gaussian noise (AWGN) samples with variance $\widetilde \sigma _{\rm{B}}^2$  and $\widetilde \sigma _{\rm{E}}^2$, respectively. That's to say, there exists $\sigma _{\rm{B}}^2{\rm{ = }}\beta _{\rm{E}}^{ - 1}\widetilde \sigma _{\rm{B}}^2$ and $\sigma _{\rm{E}}^2 = \beta _{\rm{E}}^{ - 1}\widetilde \sigma _{\rm{E}}^2$.  Also we have the following assumption:
\begin{enumerate}
\item Channels obey the block-fading model where the channel state remains quasi-static  over each coherence interval (or channel block), but becomes independent across a different fading block.
\item $T$  within each coherence interval is large enough to allow for invoking random coding arguments~\cite{Yingbin_Liang}.
\item ${{\bf{h}}_{{\rm{B}}}}$ is accurately estimated by Bob and is also known by Alice using a noiseless feedback link from Bob and  the knowledge of both ${{\bf{h}}_{{\rm{B}}}}$ and ${{\bf{H}}_{{\rm{E}}}}$ is available at Eve~\cite{Negi}.
\end{enumerate}
\label{system model}
\begin{figure}[!t]
\centering \includegraphics[width=1.0\linewidth]{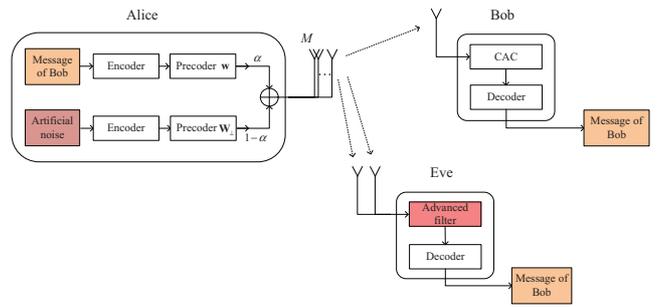}
\caption{Illustration of the security threat in artificial-noise-aided $M$-1-2 wiretap channels: An eavesdropper intends to exploit advanced receiver filters to enhance the interception  of  information belonging to Bob.}\label{fig:sys}
\vspace{-8pt}
\end{figure}
As shown in Fig.~\ref{fig:sys}, the key idea of AN proposed in~\cite{Negi} is  that Alice  combines the  information-bearing symbols  intended for  Bob  with the AN symbols   which are generated  into the null space of  Bob's channel. Thanks to  AN, the signal reception of Eve can be disturbed while  Bob is not affected. Thus positive secrecy capacity can be achieved and  huge secrecy rate gains can be maintained.  In this case, the transmitted signal vector ${{\bf{s}}_i}$ at Alice  is  represented by  ${{\bf{s}}_i} = {{\bf{v}}_i} + {{\bf{n}}_i}, i = 1, \ldots ,T$ where ${{\bf{v}}_i}\in {{\mathbb C}^{{M} \times 1}}$ is  the transmitted  signal vector  intended for Bob and ${{\bf{n}}_i}\in {{\mathbb C}^{{M} \times 1}}$ is the AN vector.  In particular,  ${{\bf{v}}_i}$ satisfies  ${{\bf{v}}_i} = {\bf{w}}{x_i}$ where ${\bf{w}}\in {{\mathbb C}^{{M} \times 1}}$ is the beamforming vector and ${x_i}$ is the  the information bearing symbol. AN, independent of $ {{\bf{v}}_i}$, satisfies ${{\bf{n}}_i} = {{\bf{W}}_ \bot }{{\bf{a}}_i}$ where ${{\bf{W}}_ \bot }\in {{\mathbb C}^{{M} \times \left( {{M} - 1} \right)}}$ lies in the null space of ${{\bf{H}}_{\rm{B}}}$.  It is assumed that  ${x_i}$ and elements of  ${{\bf{a}}_i}$ are complex Gaussian distributed and independently transmitted  in time domain, with variance $\sigma _x^2$ and  $\sigma _a^2$ respectively.  Then the  symbols received at Bob and Eve respectively  become
\begin{equation}
\begin{array}{l}
{{\bf{y}}_{{\rm{B}},i}} = {{\bf{h}}_{\rm{B}}}{\bf{w}}{x_i} + {{\bf{n}}_{{\rm{B}},i}},i = 1, \ldots , T\\
{{\bf{y}}_{{\rm{E}},i}} = {{\bf{H}}_{\rm{E}}}{\bf{w}}{x_i} + {{\bf{H}}_{\rm{E}}}{{\bf{W}}_ \bot }{{\bf{a}}_i} + {{\bf{n}}_{{\rm{E}},i}},i = 1, \ldots , T
\end{array}
\end{equation}
We consider that the total power per transmission denoted by $P$  satisfies   $P = \sigma _x^2 + \left( {{M} - 1} \right)\sigma _a^2$. We  also denote the fraction of total power allocated to the information bearing signal as $\alpha $ and have:
\begin{equation}
\begin{array}{l}
\sigma _x^2 = \alpha P\\
\sigma _a^2 = {{\left( {1 - \alpha } \right)P} \mathord{\left/
 {\vphantom {{\left( {1 - \alpha } \right)P} {\left( {{M} - 1} \right)}}} \right.
 \kern-\nulldelimiterspace} {\left( {{M} - 1} \right)}}, \alpha  \in \left[ {0,1} \right]
\end{array}
\end{equation}
\section{Novel Eavesdropping  Threat and Its Technical Modeling}
\label{NETTM}
Based on  above AN based context, we in this section introduce a novel  statistical eavesdropping threat and give the specific modeling  details for this  threat, following  the analysis of secrecy capacity  under  traditional receiver assumption for Eve.

\subsection{What We Originally Thought:  Eavesdropper Suffers from AN Interference}
To begin with, let us focus on the traditional definition of secrecy capacity which is  based on the maximization of  mutual information difference. Note that the capacity-achievable combiner (CAC) at  Eve  is not well specified.  When AN is  introduced for Eve, the same methodology for secrecy capacity  is  followed since  Eve does not know the specific value  of  AN and can deem  AN as  a colored noise equivalently. In this context,  the secrecy capacity and its lower bound  can be respectively represented by~\cite{Negi}:

\begin{equation}
\begin{array}{l}
{C_{\rm{S}}} = \mathop {\max }\limits_{\sigma _x^2,\sigma _a^2} {\left[ {{C_{\rm{B}}} - {C_{\rm{E}}}} \right]^ + }\\
{C_{{\rm{S,lower}}}} = \mathop {\max }\limits_{\sigma _x^2,\sigma _a^2} {\left[ {{C_{\rm{B}}} - {C_{{\rm{E,upper}}}}} \right]^ + }
\end{array}
\end{equation}
where ${\left[ x \right]^ + } = \max \left\{ {0,x} \right\}$, and
\begin{equation}
{C_{\rm{B}}} =  {{{\log }_2}\left( {1 + {\frac{{\sigma _x^2}}{{\sigma _{\rm{B}}^2}}}{\left\| {{{\bf{h}}_{\rm{B}}}} \right\|^2}} \right)},
\end{equation}
\begin{equation}
{C_{\rm{E}}} = {\log _2}\left( {1 + \sigma _{{x}}^2{\bf{w}}^{\rm{H}}{\bf{H}}_{\rm{E}}^{\rm{H}}{\bf{K}}_{{\rm{E,1}}}^{ - 1}{{\bf{H}}_{\rm{E}}}{{\bf{w}}}} \right),
\end{equation}
and
\begin{equation}
{C_{{\rm{E,upper}}}} = {\log _2}\left( {1 + \frac{{{M} - 1}}{{{\alpha ^{ - 1}} - 1}}{\bf{w}}^{\rm{H}}{\bf{H}}_{\rm{E}}^{\rm{H}}{\bf{K}}_{{\rm{E,2}}}^{ - 1}{{\bf{H}}_{\rm{E}}}{{\bf{w}}}} \right)
\end{equation}
where $\bf{w}$ is the  beamformer matrix  satisfying ${{{\bf{w}} = {{\bf{h}}_{\rm{B}}}} \mathord{\left/
 {\vphantom {{{\bf{w}} = {{\bf{h}}_{\rm{B}}}} {\left\| {{{\bf{h}}_{\rm{B}}}} \right\|}}} \right.
 \kern-\nulldelimiterspace} {\left\| {{{\bf{h}}_{\rm{B}}}} \right\|}}$.
 ${{\bf{K}}_{{\rm{E,1}}}} = \sigma _{\rm{a}}^2{{\bf{H}}_{\rm{E}}}{{\bf{W}}_ \bot }{\bf{W}}_ \bot ^{\rm{H}}{\bf{H}}_{\rm{E}}^{\rm{H}} + {{\bf{I}}_2}$ is derived under  the noisy Eve case  while  ${{\bf{K}}_{{\rm{E,2}}}}{\rm{ = }}{{\bf{H}}_{\rm{E}}}{{\bf{W}}_ \bot }{\bf{W}}_ \bot ^{\rm{H}}{\bf{H}}_{\rm{E}}^{\rm{H}}$ is formulated  in the case of a noiseless eavesdropping environment which creates an lower bound on the secrecy capacity~\cite{Zhou}.

However,  the continual utilization of  CAC is not feasible for Eve because it can exploit the statistical property of received data, not even necessarily distinguish between signals,  and  enhance the decoding efficiency of the information symbols that was intercepted.
\subsection{What Eve Can Actually Do: Introducing Statistical Filter  to Combat AN Interference }
In this section, we introduce the concept of OGM  eavesdropping and show how Eve can exploit the statistical  information of received signals to devise the advanced  filter. Let us consider the received symbols of length $N$ ($N \le T$) in one channel block for Eve:
\begin{equation}
{{\bf{Y}}_{\rm{E}}} = {{\bf{H}}_{\rm{E}}}{\bf{wx}} + {{\bf{H}}_{\rm{E}}}{{\bf{W}}_ \bot }{\bf{A}} + {{\bf{N}}_{\rm{E}}}
\end{equation}
where ${{\bf{Y}}_{\rm{E}}} = \left[ {{{\bf{y}}_{{\rm{E}},i \in {\cal N}}}} \right] \in {{\mathbb C}^{{2} \times N}}$, ${\bf{x}} = \left[ {{x_{i \in {\cal N}}}} \right] \in {{\mathbb C}^{1 \times N}}$, ${\bf{A}} = \left[ {{{\bf{a}}_{i \in {\cal N}}}} \right] \in {{\mathbb C}^{\left( {{M} - 1} \right) \times N}}$, ${{\bf{N}}_{\rm{E}}} = \left[ {{{\bf{n}}_{{\rm{E}},i \in {\cal N}}}} \right] \in {{\mathbb C}^{{2} \times N}}$. ${\cal N}$ represents the time sample set defined as ${\cal N} \buildrel \Delta \over = \left\{ {t\left| {1 \le t \le N} \right.} \right\}$. We consider a transformation satisfying:
\begin{equation}
{\overline {\bf{Y}} _{\rm{E}}}{\rm{ = }}\overline {\bf{U}} {{\bf{Y}}_{\rm{E}}}
\end{equation}
where ${\overline{\bf{U}}} \in {{\mathbb C}^{{2} \times 2}}$  is the unitary  matrix and can be  known and arbitrarily configured by Eve. We define ${\overline{\bf{U}}}$  as the transform matrix. Then we can rewrite the  signal model under $\alpha  \in \left( {0,1} \right)$ as follows:
\begin{equation}
{\overline {\bf{Y}} _{\rm{E}}} = {\bf{HS}} + {\overline {\bf{N}} _{\rm{E}}}
\end{equation}
where ${\bf{H}} \buildrel \Delta \over = \left[ {\begin{array}{*{20}{c}}
{{\sigma _x}\overline {\bf{U}} {{\bf{H}}_{\rm{E}}}{\bf{w}}}&{{\sigma _a}\overline {\bf{U}} {{\bf{H}}_{\rm{E}}}{{\bf{W}}_ \bot }}
\end{array}} \right]$,  ${\bf{S}} \buildrel \Delta \over = \left[ {\begin{array}{*{20}{c}}
{\overline {\bf{x}} }\\
{\overline {\bf{A}} }
\end{array}} \right]$ and ${\overline {\bf{N}} _{\rm{E}}} = \overline {\bf{U}} {{\bf{N}}_{\rm{E}}}$. Each of elements in ${\overline {\bf{x}} }$ and ${\overline {\bf{A}} }$ are respectively power-normalized version of that in $\bf{x}$ and $\bf{A}$.
Based on this model, we  give the following steps regarding  the novel OGM  filter design:
\subsubsection{Whitening  Data}
We define  the  estimation of  covariance matrix of ${\overline {\bf{Y}} _{\rm{E}}}$ as ${{\widehat{\bf R}}_{{{\overline {\bf{Y}} }_{\rm{E}}}}}{\rm{ }} \buildrel \Delta \over = {{{{\overline {\bf{Y}} }_{\rm{E}}}\overline {\bf{Y}} _{\rm{E}}^{\rm{H}}} \mathord{\left/
 {\vphantom {{{{\overline {\bf{Y}} }_{\rm{E}}}\overline {\bf{Y}} _{\rm{E}}^{\rm{H}}} N}} \right.
 \kern-\nulldelimiterspace} N}$ and derive the noiseless version as follows:
\begin{equation}
{{{{\widehat {\bf R}}}}_{\rm S}} = {{{{\widehat {\bf R}}}}_{{\overline {\bf{Y}} _{\rm{E}}}}} - \sigma _{\rm E}^2{{\bf{I}}_2}
\end{equation}
The  whitening operation, defined as ${\bf{Q}}\in {{\mathbb C}^{{2} \times 2}}$, is  performed on  ${\overline {\bf{Y}} _{\rm{E}}}$ to eliminate the  second-order  statistical information, which  is given as follows:
\begin{equation}
{{\bf{I}}_2}{\rm{ = }}{\bf{Q}}{{{{\widehat {\bf R}}}}_{\rm S}}{{\bf{Q}}^{\rm{H}}}
\end{equation}
After performing singular value decomposition (SVD) on ${{{{\widehat {\bf R}}}}_{\rm S}}$, we have   ${{{{\widehat {\bf R}}}}_{\rm S}} = {\bf{U\Sigma }}{{\bf{U}}^{\rm{H}}}$. Based on equation (12), we have:
\begin{equation}
\begin{array}{l}
{\bf{Q}} = {\bf{X}}{{\bf{\Sigma }}^{{{ - 1} \mathord{\left/
 {\vphantom {{ - 1} 2}} \right.
 \kern-\nulldelimiterspace} 2}}}{{\bf{U}}^{\rm{H}}},
{\bf{X}}{{\bf{X}}^{\rm{H}}} = {{\bf{I}}_2}
\end{array}
\end{equation}
where ${\bf{X}}\in {{\mathbb C}^{{2} \times 2}}$ is an unknown matrix for further optimization. Then the signal after whitening is given as:
\begin{equation}\label{E.15}
{{\bf{Y}}_\Delta }\buildrel \Delta \over = {\bf{QY}}{\rm{ = }}{\bf{X}}{\overline {\bf{V}}}{\bf{S}}{\rm{ + }}{\bf{Q}}{\overline {\bf{N}} _{\rm{E}}}
\end{equation}
where ${\overline {\bf{V}}}$ satisfies $\overline {\bf{V}}  = {{\bf{\Sigma }}^{ - 1/2}}\left[ {\begin{array}{*{20}{c}}
{{{\bf{\Sigma }}^{1/2}}}&{\bf{0}}
\end{array}} \right]{{\bf{V}}^{\rm{H}}}$. ${{\bf{V}}^{\rm{H}}}$ is the right eigenmatrix satisfying the  matrix SVD  with ${\bf{H}}{\rm{ = }}{\bf{U}}\left[ {\begin{array}{*{20}{c}}
{{{\bf{\Sigma }}^{1/2}}}&{\bf{0}}
\end{array}} \right]{{\bf{V}}^{\rm{H}}}$. As we can see, any replacement  for the matrix ${\bf{X}}$ with other  matrices satisfying equation (14) leaves  the covariance of ${{{\bf{Y}}_{\Delta}}}$ unchanged, which means that the power influence of AN  is removed. The same happens when   any replacement  for unitary matrices $\bf V$ with other unitary matrix.  However,  ${\bf{X}}{{\overline{\bf{V}}}^{\rm{H}}}$ causes a phase error on the symbol estimation of   $\bf{S}$ at Eve.


\subsubsection{Projecting the Processed Data onto a Basis}
To eliminate the phase error, Eve is designed to  project the whitened  data onto ${\bf{X}}{{\overline{\bf{V}}}^{\rm{H}}}\left( {:,1} \right)$ and to estimate  the information bearing symbols intended for  Bob.  We denote the projector as
\begin{equation}
{\bf{q}} = {\bf{X}}{{\overline{\bf{V}}}^{\rm{H}}}\left( {:,1} \right) = {\bf{QH}}\left( {:,1} \right)
\end{equation}
It is obviously that there exists ${\left\| {\bf{q}} \right\|^2} \le 1$.
Based on the overall two steps, we can have the following proposition for  the symbol-level filtering.
 \newtheorem{proposition}{Proposition}
\begin{proposition}
Given the first column of  matrix  ${\bf{H}}$,  the estimated version of the symbol vector  ${\overline {\bf{x}} }$ under the OGM filter with the region  $\alpha  \in \left( {0,1} \right)$ can be derived as follows:
\begin{equation}\label{E.20}
\widehat {\overline{\bf{x}} } = {{\bf{q}}^{\rm{H}}}{\bf{Q}}{{\overline{\bf{Y}}}_{\rm{E}}}
\end{equation}
 \end{proposition}
 The overall serial filtering  process at Eve  can be shown in Fig.~\ref{fig:2}.

\section{Construction of Optimal Grassmann Eavesdropping  Manifold}
\label{COGEM}
In this section, we formulate an optimization problem of  minimizing  MSE on the variable matrix $\bf{X}$ and derive an optimal Grassmann manifold by solving this problem. Not that we only consider the case where the power ratio satisfies $0 < \alpha < 1$ in the following. If AN does not exist ($ \alpha = 1$), the conventional CAC receiver can be utilized at Eve. The detection of whether AN happens can be based on the rank detection of the covariance matrix ${{{{\widehat {\bf R}}}}_{\rm S}}$. We do not specify the details of  detection but focus on the following filter design.
\subsection{{Optimization Formulation}}
Consider the MSE in estimating  symbols $\widehat {\overline{\bf{x}} }$ and  give the  symbol estimation error  as follows:
\begin{equation}
\Delta {\overline{\bf{x}}} \buildrel \Delta \over =\widehat {\overline {\bf x} } - \overline {\bf{x}}{\rm{ = }}{{\bf{q}}^{\rm{H}}}{\bf{Q}}{{\overline {\bf{Y}}}_{\rm{E}}} - {\bf{S}}\left( {1,:} \right)
\end{equation}
with the NMSE defined as:
\begin{equation}
{\rm{NMSE}} \buildrel \Delta \over = \mathbb{E}\left\{ {{{{{\left\| {\Delta \overline {\bf{x}} } \right\|}^2}} \mathord{\left/
 {\vphantom {{{{\left\| {\Delta \overline {\bf{x}} } \right\|}^2}} N}} \right.
 \kern-\nulldelimiterspace} N}} \right\}
\end{equation}
After simplification,  the NMSE can be derived as equation (19) if we define ${\bf{Z}}{\rm{ = }}{{\bf{X}}^{\rm{H}}}{\bf{X}}$. The optimization problem of minimizing NMSE  is built up as follows:
\begin{equation}\setcounter{equation}{20}
 \begin{aligned}
 & \underset{{\bf{Z}}}{\text{min}}
 & &  {{\rm{Tr}}\left\{ {{\bf{TZ}}\left( {{{\bf{I}}_2}{\rm{ + }}\sigma _{\rm{E}}^2{{\bf{\Sigma }}^{ - 1}}} \right){{\bf{Z}}^{\rm{H}}} - 2{\bf{T}}{{\bf{Z}}^{\rm{H}}}} \right\}}\\
 & \text{s.\,t.}
 && {\bf{T}} = {{\bf{\Sigma }}^{{{ - 1} \mathord{\left/
 {\vphantom {{ - 1} 2}} \right.
 \kern-\nulldelimiterspace} 2}}}{{\bf{U}}^{\rm{H}}}{\bf{H}}\left( {:,1} \right){\bf{H}}{\left( {:,1} \right)^{\rm{H}}}{\bf{U}}{{\bf{\Sigma }}^{{{ - 1} \mathord{\left/
 {\vphantom {{ - 1} 2}} \right.
 \kern-\nulldelimiterspace} 2}}}\\
 &&& {\bf{H}}{\left( {:,1} \right)^{\rm{H}}}{\bf{U}}{{\bf{\Sigma }}^{{{ - 1} \mathord{\left/
 {\vphantom {{ - 1} 2}} \right.
 \kern-\nulldelimiterspace} 2}}}{{\bf{Z}}^{\rm{H}}}{{\bf{\Sigma }}^{{{ - 1} \mathord{\left/
 {\vphantom {{ - 1} 2}} \right.
 \kern-\nulldelimiterspace} 2}}}{{\bf{U}}^{\rm{H}}}{\bf{H}}\left( {:,1} \right) - 1 \le 0, \\
 &&&{\rm{Tr}}\left\{ {\bf{Z}} \right\}{\rm{ = }}2 ,\\
 &&&{\bf{Z}}\succeq {\bf{0}}.\\
 \end{aligned}
 \end{equation}
 In this optimization, the second constraint is constructed according to the principle of ${\left\| {\bf{q}} \right\|^2} \le 1$. The third and fourth constraint mean that  matrix $\bf{Z}$ should satisfy the same rank requirements as ${\bf{X}}{{\bf{X}}^{\rm{H}}}$.
\subsection{{Optimal Solution}}
 We can easily prove that the above optimization  is a convex problem.  Besides, the set of  constraint  functions is convex and the Slater condition is satisfied. Therefore, the KKT conditions are sufficient and necessary for the optimum solution.
 \begin{figure*}
\begin{equation}\setcounter{equation}{19}
{\rm{NMSE}} = {\bf{H}}{\left( {:,1} \right)^{\rm{H}}}{\bf{U}}{{\bf{\Sigma }}^{{{ - 1} \mathord{\left/
 {\vphantom {{ - 1} 2}} \right.
 \kern-\nulldelimiterspace} 2}}}{\bf{Z}}\left( {{{\bf{I}}_2}{\rm{ + }}\sigma _{\rm{E}}^2{{\bf{\Sigma }}^{ - 1}}} \right){{\bf{Z}}^{\rm{H}}}{{\bf{\Sigma }}^{{{ - 1} \mathord{\left/
 {\vphantom {{ - 1} 2}} \right.
 \kern-\nulldelimiterspace} 2}}}{{\bf{U}}^{\rm{H}}}{\bf{H}}\left( {:,1} \right) - {\rm{2}}{\bf{H}}{\left( {:,1} \right)^{\rm{H}}}{\bf{U}}{{\bf{\Sigma }}^{{{ - 1} \mathord{\left/
 {\vphantom {{ - 1} 2}} \right.
 \kern-\nulldelimiterspace} 2}}}{{\bf{Z}}^{\rm{H}}}{{\bf{\Sigma }}^{{{ - 1} \mathord{\left/
 {\vphantom {{ - 1} 2}} \right.
 \kern-\nulldelimiterspace} 2}}}{{\bf{U}}^{\rm{H}}}{\bf{H}}\left( {:,1} \right){\rm{ + 1}}
\end{equation}
\vspace{-20pt}
 \end{figure*}
 \begin{figure*}
 \begin{align*}
{\cal L}\left( {{\bf{Z}},\lambda ,\mu ,{\bf{\Psi }}} \right) &={\rm{ Tr}}\left\{ {{{\bf{\Sigma }}^{{{ - 1} \mathord{\left/
 {\vphantom {{ - 1} 2}} \right.
 \kern-\nulldelimiterspace} 2}}}{{\bf{U}}^{\rm{H}}}{\bf{H}}\left( {:,1} \right){\bf{H}}{{\left( {:,1} \right)}^{\rm{H}}}{\bf{U}}{{\bf{\Sigma }}^{{{ - 1} \mathord{\left/
 {\vphantom {{ - 1} 2}} \right.
 \kern-\nulldelimiterspace} 2}}}{\bf{Z}}\left( {{{\bf{I}}_2}{\rm{ + }}\sigma _{\rm{E}}^2{{\bf{\Sigma }}^{ - 1}}} \right){{\bf{Z}}^{\rm{H}}} - 2{{\bf{\Sigma }}^{{{ - 1} \mathord{\left/
 {\vphantom {{ - 1} 2}} \right.
 \kern-\nulldelimiterspace} 2}}}{{\bf{U}}^{\rm{H}}}{\bf{H}}\left( {:,1} \right){\bf{H}}{{\left( {:,1} \right)}^{\rm{H}}}{\bf{U}}{{\bf{\Sigma }}^{{{ - 1} \mathord{\left/
 {\vphantom {{ - 1} 2}} \right.
 \kern-\nulldelimiterspace} 2}}}{{\bf{Z}}^{\rm{H}}}} \right\}\\
  &- \lambda \left\{ {{\rm{Tr}}\left\{ {{{\bf{\Sigma }}^{{{ - 1} \mathord{\left/
 {\vphantom {{ - 1} 2}} \right.
 \kern-\nulldelimiterspace} 2}}}{{\bf{U}}^{\rm{H}}}{\bf{H}}\left( {:,1} \right){\bf{H}}{{\left( {:,1} \right)}^{\rm{H}}}{\bf{U}}{{\bf{\Sigma }}^{{{ - 1} \mathord{\left/
 {\vphantom {{ - 1} 2}} \right.
 \kern-\nulldelimiterspace} 2}}}{{\bf{Z}}^{\rm{H}}}} \right\}{\rm{ - }}1} \right\} - \mu \left\{ {{\rm{Tr}}\left\{ {{{\bf{Z}}^{\rm{H}}}} \right\} - {\rm{2}}} \right\} - {\rm{Tr}}\left\{ {{\bf{\Psi }}{{\bf{Z}}^{\rm{H}}}} \right\}
 \tag{21}
  \end{align*}
  \vspace{-15pt}
  \end{figure*}
  \begin{figure*}[ht]
  \begin{equation}\setcounter{equation}{23}
\lambda^{*} {\rm{ = }}\frac{{1 - {\bf{H}}{{\left( {:,1} \right)}^{\rm{H}}}{\bf{U}}{{\bf{\Sigma }}^{{{ - 1} \mathord{\left/
 {\vphantom {{ - 1} 2}} \right.
 \kern-\nulldelimiterspace} 2}}}\Gamma \left( {{{\bf{\Sigma }}^{{{ - 1} \mathord{\left/
 {\vphantom {{ - 1} 2}} \right.
 \kern-\nulldelimiterspace} 2}}}{{\bf{U}}^{\rm{H}}}{\bf{H}}\left( {:,1} \right){\bf{H}}{{\left( {:,1} \right)}^{\rm{H}}}{\bf{U}}{{\bf{\Sigma }}^{{{ - 1} \mathord{\left/
 {\vphantom {{ - 1} 2}} \right.
 \kern-\nulldelimiterspace} 2}}}} \right){{\left[ {{{\bf{I}}_2} + \sigma _{\rm{E}}^2{{\bf{\Sigma }}^{ - 1}}} \right]}^{ - 1}}{{\bf{\Sigma }}^{{{ - 1} \mathord{\left/
 {\vphantom {{ - 1} 2}} \right.
 \kern-\nulldelimiterspace} 2}}}{{\bf{U}}^{\rm{H}}}{\bf{H}}\left( {:,1} \right)}}{{{\bf{H}}{{\left( {:,1} \right)}^{\rm{H}}}{\bf{U}}{{\bf{\Sigma }}^{{{ - 1} \mathord{\left/
 {\vphantom {{ - 1} 2}} \right.
 \kern-\nulldelimiterspace} 2}}}{{\left[ {{{\bf{I}}_2} + \sigma _{\rm{E}}^2{{\bf{\Sigma }}^{ - 1}}} \right]}^{ - 1}}{{\bf{\Sigma }}^{{{ - 1} \mathord{\left/
 {\vphantom {{ - 1} 2}} \right.
 \kern-\nulldelimiterspace} 2}}}{{\bf{U}}^{\rm{H}}}{\bf{H}}\left( {:,1} \right)}}-2
\end{equation}
\vspace{-15pt}
\end{figure*}
Then we build up the Lagrangian as the equation (21) and  the KKT conditions for the optimization problem  are given by:
 \begin{align*}
{{\bf{\Sigma }}^{{{ - 1} \mathord{\left/
 {\vphantom {{ - 1} 2}} \right.
 \kern-\nulldelimiterspace} 2}}}{{\bf{U}}^{\rm{H}}}{\bf{H}}\left( {:,1} \right){\bf{H}}{\left( {:,1} \right)^{\rm{H}}}{\bf{U}}{{\bf{\Sigma }}^{{{ - 1} \mathord{\left/
 {\vphantom {{ - 1} 2}} \right.
 \kern-\nulldelimiterspace} 2}}}{\bf{Z}}\left( {{{\bf{I}}_2} + \sigma _{\rm{E}}^2{{\bf{\Sigma }}^{ - 1}}} \right)\,\,\,\,\,\,\,\,\,\,\,\,\,\,\,\,\,\\
 - \left( {\lambda  + 2} \right){{\bf{\Sigma }}^{{{ - 1} \mathord{\left/
 {\vphantom {{ - 1} 2}} \right.
 \kern-\nulldelimiterspace} 2}}}{{\bf{U}}^{\rm{H}}}{\bf{H}}\left( {:,1} \right){\bf{H}}{\left( {:,1} \right)^{\rm{H}}}{\bf{U}}{{\bf{\Sigma }}^{{{ - 1} \mathord{\left/
 {\vphantom {{ - 1} 2}} \right.
 \kern-\nulldelimiterspace} 2}}}\,\,\,\,\,\,\,\,\,\,\,\,\,\,\,\,\,\\
  - \mu {{\bf{I}}_2}- {\bf{\Psi }} = 0 \,\,\,\text{(22\,a)}\\
\lambda \left\{ {{\bf{H}}{{\left( {:,1} \right)}^{\rm{H}}}{\bf{U}}{{\bf{\Sigma }}^{{{ - 1} \mathord{\left/
 {\vphantom {{ - 1} 2}} \right.
 \kern-\nulldelimiterspace} 2}}}{{\bf{Z}}^{\rm{H}}}{{\bf{\Sigma }}^{{{ - 1} \mathord{\left/
 {\vphantom {{ - 1} 2}} \right.
 \kern-\nulldelimiterspace} 2}}}{{\bf{U}}^{\rm{H}}}{\bf{H}}\left( {:,1} \right) - 1} \right\} = 0\,\,\,\text{(22\,b)}\\
 {\bf{H}}{\left( {:,1} \right)^{\rm{H}}}{\bf{U}}{{\bf{\Sigma }}^{{{ - 1} \mathord{\left/
 {\vphantom {{ - 1} 2}} \right.
 \kern-\nulldelimiterspace} 2}}}{{\bf{Z}}^{\rm{H}}}{{\bf{\Sigma }}^{{{ - 1} \mathord{\left/
 {\vphantom {{ - 1} 2}} \right.
 \kern-\nulldelimiterspace} 2}}}{{\bf{U}}^{\rm{H}}}{\bf{H}}\left( {:,1} \right) - 1 \le 0\,\,\,\text{(22\,c)}\\
 {\rm{Tr}}\left\{ {{{\bf{Z}}^{\rm{H}}}} \right\} - 2= 0\,\,\,\text{(22\,d)}\\
{\rm{Tr}}\left\{ {{\bf{\Psi }}{{\bf{Z}}^{\rm{H}}}} \right\}= 0\,\,\,\text{(22\,e)}\\
{\bf{Z}} \succeq {\bf{0}}\,\,\,\text{(22\,f)}\\
{\bf{\Psi }} \succeq {\bf{0}},{{\lambda }} \ge {{0}}, \mu \ge 0\,\,\text{(22\,g)}
\end{align*}
 \begin{figure}[!t]
\centering \includegraphics[width=1.0\linewidth]{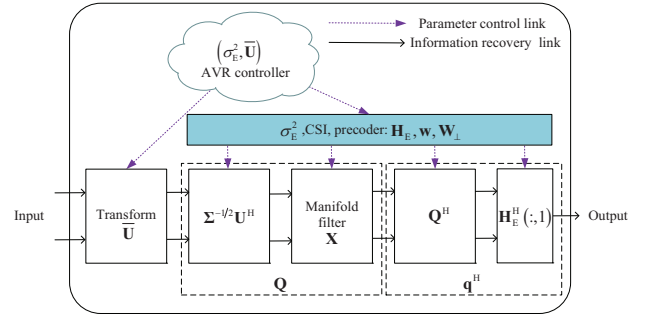}
\caption{Architecture of the proposed  OGM  filter at Eve.  }\label{fig:2}
\vspace{-8pt}
\end{figure}
According to the equation (22\,d) and (22\,e), the solution ${\bf{\Psi }}^{*}$ should be equal to be zero matrix. Considering  the rank  constraints  on both sides of the equation (22\,a), we can obtain the solution $ {\mu}^{*} {\rm{ = }}0$. Then we can get the closed form of ${\lambda}^{*}$ in equation (23) and then derive the solution ${\bf{Z}}^{*}$ as follows:
\begin{equation}
{{\bf{Z}}^{{\rm{*}}}}{\rm{ = }}\left\{ {\Gamma {\rm{ + }}\left( {\lambda^{*}  + 2} \right) \otimes {\bf{I}}} \right\}{\left[ {{{\bf{I}}_2} + \sigma _{\rm{E}}^2{{\bf{\Sigma }}^{ - 1}}} \right]^{ - 1}}
\end{equation}
where $\bf{\Gamma }$ is  the null space of matrix ${{{\bf{\Sigma }}^{{{ - 1} \mathord{\left/
 {\vphantom {{ - 1} 2}} \right.
 \kern-\nulldelimiterspace} 2}}}{{\bf{U}}^{\rm{H}}}{\bf{H}}\left( {:,1} \right){\bf{H}}{{\left( {:,1} \right)}^{\rm{H}}}{\bf{U}}{{\bf{\Sigma }}^{{{ - 1} \mathord{\left/
 {\vphantom {{ - 1} 2}} \right.
 \kern-\nulldelimiterspace} 2}}}}$.
Based on the equation (13) and (24),  a Grassmann manifold  ${\cal G}\left( {2,{\rm{2}}} \right)$ can be thus  constructed as the following:
 \begin{equation}
{\cal G}\left( {2,{\rm{2}}} \right) = \left\{ {{\bf{X}} \in {{\mathbb{C}}^{2 \times 2}}:{\bf{X}}{{\bf{X}}^{\rm{H}}} = {{{{\bf{I}}_2}}},{{\bf{X}}^{\rm{H}}}{\bf{X}} = {{\bf{Z}}^{{*}}}} \right\}
\end{equation}
As a special case where the null space is a zero matrix,  we can derive a simplified version  of  ${{\bf{Z}}^{{*}}}$ as follows:
 \begin{equation}
{\overline {\bf{Z}}^{{*}}}{\rm{ = }}\frac{{{{\left[ {{{\bf{I}}_2} + \sigma _{\rm{E}}^2{{\bf{\Sigma }}^{ - 1}}} \right]}^{ - 1}}}}{{{\bf{H}}{{\left( {:,1} \right)}^{\rm{H}}}{\bf{U}}{{\bf{\Sigma }}^{{{ - 1} \mathord{\left/
 {\vphantom {{ - 1} 2}} \right.
 \kern-\nulldelimiterspace} 2}}}{{\left[ {{{\bf{I}}_2} +\sigma _{\rm{E}}^2{{\bf{\Sigma }}^{ - 1}}} \right]}^{ - 1}}{{\bf{\Sigma }}^{{{ - 1} \mathord{\left/
 {\vphantom {{ - 1} 2}} \right.
 \kern-\nulldelimiterspace} 2}}}{{\bf{U}}^{\rm{H}}}{\bf{H}}\left( {:,1} \right)}}
\end{equation}
\section{Secrecy Performance Analysis}
\label{SPAADE}
In this section, we firstly derive  the exact  expression  of  SINR   in a closed form for  Eve under the OGM filter.
 Then we analyse the MASR  under different noise level of Eve  and show how Eve can  achieve the eavesdropping disaster: \emph{zero MASR}.
 \subsection{Closed Form of SINR Expression}
Theoretically, we assume  the data length is enough large and the covariance matrix  ${\widehat {\bf{R}}_{\rm{S}}}$ can be precisely approximated as $ {\bf{H}}{{\bf{H}}^{\rm{H}}}$. By using the OGM filter, the SINR of Eve can be derived as the following:
 \begin{equation}
{\rm{SIN}}{{\rm{R}}_{\rm{E}}} = \frac{{{{\left[ {{\bf{H}}{{\left( {:,1} \right)}^{\rm{H}}}{{\left( {{\bf{H}}{{\bf{H}}^{\rm{H}}} +\sigma _{\rm{E}}^2{{\bf{I}}_2}} \right)}^{ - 1}}{\bf{H}}\left( {:,1} \right)} \right]}^2}}}{{\sigma _{\rm{E}}^2{\bf{H}}{{\left( {:,1} \right)}^{\rm{H}}}{{\left( {{\bf{H}}{{\bf{H}}^{\rm{H}}} + \sigma _{\rm{E}}^2{{\bf{I}}_2}} \right)}^{ - 2}}{\bf{H}}\left( {:,1} \right)}}
 \end{equation}
For simplicity, we define the following equation:
 \begin{equation}
 \begin{array}{l}
{h_1} \buildrel \Delta \over = \overline {\bf{U}} \left( {1,:} \right){{\bf{H}}_{\rm{E}}}{\bf{w}},{h_2} \buildrel \Delta \over = \overline {\bf{U}} \left( {2,:} \right){{\bf{H}}_{\rm{E}}}{\bf{w}}\\
{{\bf{h}}_3} \buildrel \Delta \over = \overline {\bf{U}} \left( {1,:} \right){{\bf{H}}_{\rm{E}}}{{\bf{W}}_ \bot },{{\bf{h}}_4} \buildrel \Delta \over = \overline {\bf{U}} \left( {2,:} \right){{\bf{H}}_{\rm{E}}}{{\bf{W}}_ \bot }
\end{array}
 \end{equation}
Then  the exact expression of SINR at Eve can be expressed as:
  \begin{equation}
  {\rm{SIN}}{{\rm{R}}_{\rm{E}}} = \frac{{\Psi _1^2}}{{\sigma _{\rm{E}}^2\left( {{\Psi _2} + {\Psi _3}} \right)}}
 \end{equation}
 where ${\Psi _1}$, ${\Psi _2}$ and ${\Psi _3}$ can be shown in equation (30). Based on this equation, we propose  the following theorem:
   \newtheorem{theorem}{Theorem}
  \begin{theorem}
  For the $M$-1-2 Gaussian wiretap channel, the MASR  is given by:
 \begin{equation}\setcounter{equation}{31}
 \mathop {\max }\limits_{\sigma _x^2,\sigma _a^2,\sigma _{\rm{B}}^2}\, \mathop {\min }\limits_{\sigma _{\rm{E}}^2,\overline {\bf{U}} } {C_S}
 \end{equation}
 where
 \begin{equation}
 {C_S} = \left\{ {\begin{array}{*{20}{c}}
{C_S^ - }&{{\rm{0}} < \sigma _x^2 < P}\\
{C_S^ + }&{\sigma _x^2 = 0}\,\, {\rm or}\,\, {\sigma _x^2 = P}
\end{array}} \right.
 \end{equation}
 with
  \begin{equation}
  C_{\rm{S}}^ -  = {\left[ {{{\log }_2}\left( {\frac{{1 + \frac{{\sigma _x^2}}{{\sigma _{\rm{B}}^2}}{{\left\| {{{\bf{h}}_{\rm{B}}}} \right\|}^2}}}{{1 + {\rm{SIN}}{{\rm{R}}_{\rm{E}}}}}} \right)} \right]^ + }
   \end{equation}
   and
    \begin{equation}
C_S^ +  ={\left[ {\log _2}\left( {\frac{{1 + \frac{{\sigma _x^2}}{{\sigma _{\rm{B}}^2}}{{\left\| {{{\bf{h}}_{\rm{B}}}} \right\|}^2}}}{{1 + \frac{{\sigma _x^2}}{{\sigma _{\rm{E}}^2}}{{\left\| {{{{{\bf{H}}_{\rm{E}}}{{\bf{h}}_{\rm{B}}}} \mathord{\left/
 {\vphantom {{{{\bf{H}}_{\rm{E}}}{{\bf{h}}_{\rm{B}}}} {\left\| {{{\bf{h}}_{\rm{B}}}} \right\|}}} \right.
 \kern-\nulldelimiterspace} {\left\| {{{\bf{h}}_{\rm{B}}}} \right\|}}} \right\|}^2}}}} \right) \right]^ + }
   \end{equation}
 \end{theorem}
 When AN exists or equivalently $0<\alpha<1$, OGM filter can be generated  as shown in the above section. In this case, the achievable secrecy rate can be easily expressed as $ C_{\rm{S}}^ -$. Otherwise,  Eve adopts the conventional CAC receiver and the achievable secrecy rate can be expressed as $ C_{\rm{S}}^ +$.  Note that the minimization operation is considered  since ${\sigma _{\rm{E}}^2}$, ${\overline {\bf{U}}}$ can be optimized by Eve  to reduce the secrecy rate.
 In what follows, we will discuss the effect of  equivalent noise levels at  Eve on the MASR performance.
\subsection{Noiseless Case}
 \begin{figure*}[ht]
\begin{equation}\setcounter{equation}{30}
\begin{array}{l}
{\Psi _1} = \sigma _x^2\sigma _a^2A +\sigma _x^2\sigma _{\rm{E}}^2B,{\Psi _2} = \sigma _x^2\sigma _a^4C,{\Psi _3} = 2\sigma _x^2\sigma _{\rm{E}}^4B + 2\sigma _x^2\sigma _a^2\sigma _{\rm{E}}^2D,\\
A = {\left| {{h_1}} \right|^2}{\left\| {{{\bf{h}}_4}} \right\|^2} + {\left| {{h_2}} \right|^2}{\left\| {{{\bf{h}}_3}} \right\|^2} - 2{\rm{Re}}\left\{ {h_1^{\rm{H}}{h_2}{\bf{h}}_4^{\rm{H}}{{\bf{h}}_3}} \right\},B = {\left| {{h_1}} \right|^2} + {\left| {{h_2}} \right|^2} - 2{\rm{Re}}\left\{ {h_1^{\rm{H}}{h_2}} \right\},\\
C = {\left\| {{{\bf{h}}_4}} \right\|^4}{\left| {{h_1}} \right|^2} + {\left\| {{{\bf{h}}_3}} \right\|^4}{\left| {{h_2}} \right|^2} + \left( {{{\left| {{h_1}} \right|}^2} + {{\left| {{h_2}} \right|}^2}} \right){\left\| {{{\bf{h}}_3}} \right\|^2}{\left\| {{{\bf{h}}_4}} \right\|^2} - 2\left( {{{\left\| {{{\bf{h}}_3}} \right\|}^2} + {{\left\| {{{\bf{h}}_4}} \right\|}^2}} \right){\rm{Re}}\left\{ {h_1^{\rm{H}}{h_2}{\bf{h}}_4^{\rm{H}}{{\bf{h}}_3}} \right\},\\
D = {\left| {{h_1}} \right|^2}{\left\| {{{\bf{h}}_4}} \right\|^2} + {\left| {{h_2}} \right|^2}{\left\| {{{\bf{h}}_3}} \right\|^2} - 2{\mathop{\rm Re}\nolimits} \{ h_1^{\rm{H}}{h_2}{\bf{h}}_4^{\rm{H}}{{\bf{h}}_3}\}  + \left( {{{\left| {{h_1}} \right|}^2} + {{\left| {{h_2}} \right|}^2}} \right){\rm{Re}}\left\{ {{\bf{h}}_4^{\rm{H}}{{\bf{h}}_3}} \right\} - \left( {{{\left\| {{{\bf{h}}_3}} \right\|}^2} + {{\left\| {{{\bf{h}}_4}} \right\|}^2}} \right){\rm{Re}}\left\{ {h_1^{\rm{H}}{h_2}} \right\}.
\end{array}
\end{equation}
\vspace{-15pt}
\end{figure*}

In this case, the equivalent noise level is arbitrary small. An intuitive interpretation is that Bob or  Eve is much closer to Alice~\cite{Negi} and/ or have very low noise level. In this context, we have the following theorem:
  \begin{theorem}
 With $\sigma _{\rm{E}}^2 \to 0$, the MASR is equal to zero.
\end{theorem}
  \begin{proof}
let us consider the case where ${\rm{0}} < {\sigma _x^2 < P}$. The  SINR of Eve will satisfy :
  \begin{equation}\setcounter{equation}{35}
  \frac{{{\rm{SIN}}{{\rm{R}}_{\rm{E}}}\sigma _{\rm{E}}^2}}{{\sigma _x^2}} \to {\frac{{{A^2}}}{C}}
  \end{equation}
The secrecy rate is  given by:
  \begin{equation}
 C_{\rm{S}}^ -  = \left[ {{{\log }_2}\left( {\frac{{1 + \frac{{\sigma _x^2}}{{\sigma _{\rm{B}}^2}}{{\left\| {{{\bf{h}}_{\rm{B}}}} \right\|}^2}}}{{1 + \frac{{\sigma _x^2}}{{\sigma _{\rm{E}}^2}}\frac{{{A^2}}}{C}}}} \right)} \right]_{{\rm{0}} < \sigma _x^2 < P}^{\rm{ + }},
   \end{equation}
 Note that ${C_{\rm{S}}^ -}$ increases monotonically with ${\sigma _x^2}$  when $\frac{{{{\left\| {{{\bf{h}}_{\rm{B}}}} \right\|}^2}}}{{\sigma _{\rm{B}}^{\rm{2}}}} > \frac{{{A^{\rm{2}}}}}{{\sigma _{\rm{E}}^{\rm{2}}C}}$, and decreases monotonically, otherwise. Therefore, the MASR is achieved at ${\sigma _x^2}=P$.  At this point, the achievable  secrecy rate is transformed into ${C_{\rm{S}}^ +}$. Obviously,  ${C_{\rm{S}}^ +}$ is equal to be zero when  $\sigma _{\rm{E}}^2 \to 0$. The proof is completed.



\end{proof}
\subsection{Noisy Case}
In this situation,  the equivalent noise can not be ignored at both  Bob and Eve whose noise level and distance to Alice can be  arbitrarily managed.  Besides,  it is noted that  the variables, such as ${h_1}$,${h_2}$,${{\bf H}_3}$, and ${{\bf H}_4}$, in ${C_{\rm{S}}^ -}$,  depend on  Eve's CSI,  Alice's beamformer and $\overline {\bf{U}} $.  As a feasible assumption, Eve is enabled to obtain those  available information. Especially,   ${\sigma _{\rm{E}}^2}$ and  ${\overline {\bf{U}}}$ are both  controllable by Eve.  Therefore,  the resulted SINR of Eve can be modified and  redesigned by itself. Then, we have the following proposition:
 \begin{proposition}
 Given perfect CSI  and Alice's beamformer at noisy Eve, there always exist  variables   ${\sigma _{\rm{E}}^2}$, ${\overline {\bf{U}}}$ such that ${C_{\rm{S}}^ -}$ can be restricted to be  zero; Particularly,
 ${C_{\rm{S}}^ -}$ at high SINR is equal to zero $\forall P \in \left( {0,\infty } \right)$  if  the following sufficient condition is satisfied:
 \begin{equation}
\sigma _{\rm{E}}^{\rm{2}}\left( {\frac{{{B^{\rm{2}}}}}{{{A^{\rm{2}}}}} - \frac{{2B}}{C}} \right) \le 2\left( {\frac{B}{A}- \frac{D}{C}} \right), \frac{B}{A} < \frac{D}{C},
\end{equation}
\begin{equation}
\frac{{\sigma _{\rm{E}}^2}}{{\sigma _{\rm{B}}^2}} < \frac{{{{\left\| {{{\bf{H}}_{\rm{E}}}{{{\bf{h}}_{\rm{B}}^{\rm{H}}} \mathord{\left/
 {\vphantom {{{\bf{h}}_{\rm{B}}^{\rm{H}}} {\left\| {{{\bf{h}}_{\rm{B}}}} \right\|}}} \right.
 \kern-\nulldelimiterspace} {\left\| {{{\bf{h}}_{\rm{B}}}} \right\|}}} \right\|}^2}}}{{{{\left\| {{{\bf{h}}_{\rm{B}}}} \right\|}^2}}}
 \end{equation}
whichs make log function  to be monotonically increasing with  $\sigma _{{x}}^{\rm{2}}$, or
 \begin{equation}
 \sigma _{\rm{E}}^{\rm{2}}\left( {\frac{{{B^{\rm{2}}}}}{{{A^{\rm{2}}}}} - \frac{{2B}}{C}} \right) \ge 2\left( {\frac{B}{A}-\frac{D}{C}} \right),\frac{B}{A}>\frac{D}{C}
 \end{equation}
which makes log function  to be monotonically decreasing with  $\sigma _{{x}}^{\rm{2}}$, or
 \begin{equation}
 \frac{{{A^{\rm{2}}}}}{{\sigma _{\rm{E}}^{\rm{2}}C}} \ge \frac{{{{\left\| {{{\bf{h}}_{\rm{B}}}} \right\|}^2}}}{{\sigma _{\rm{B}}^{\rm{2}}}},\frac{D}{C} = \frac{B}{A}
 \end{equation}
which makes log function to be irrelevant with the transmission power. The values of $A$, $B$, $C$, $D$ can be shown in equation (30).
 \end{proposition}
 \begin{proof}
 Let us focus on the  monotonicity of the inner log function in  ${C_{\rm{S}}^ -}$. According to its  high-order derivative, we can prove that there exists available variables ${\sigma _{\rm{E}}^2}$, ${\overline {\bf{U}}}$ such that the log function  is monotone. However,  it is exhausting to obtain a sufficient and necessary condition in the form of an explicit expression of those variables. Particularly in high SINR and $0<\alpha<1$, we transform the original ${C_{\rm{S}}^ -}$ into  $C_{\rm{S}}^ -  = {\left[ {{{\log }_2}\left( {\frac{{\sigma _{\rm{E}}^2}}{{\sigma _{\rm{B}}^2}}{{\left\| {{{\bf{h}}_{\rm{B}}}} \right\|}^2}\frac{{\sigma _a^4C - 2\sigma _a^2\sigma _{\rm{E}}^2D + 2\sigma _{\rm{E}}^4B}}{{\sigma _a^4{A^2} - 2\sigma _{\rm{E}}^2AB\sigma _a^2 + \sigma _{\rm{E}}^4{B^2}}}} \right)} \right]^ + }$. Now we introduce the sufficient condition to achieve zero secrecy rate.

 When the log function  is controlled by Eve  to be  monotonically increasing with $\alpha$,  the condition (37), (38) has to be satisfied.  In this situation, the
  achievable secrecy rate is gradually increased with the increase of transmission  power allocated for information bearing signals and the maximum will be achieved at $\alpha=1$. Thus the condition  $\frac{{\sigma _{\rm{E}}^2}}{{\sigma _{\rm{B}}^2}} < \frac{{{{\left\| {{{\bf{H}}_{\rm{E}}}{{{{\bf{h}}_{\rm{B}}}} \mathord{\left/
 {\vphantom {{{{\bf{h}}_{\rm{B}}}} {\left\| {{{\bf{h}}_{\rm{B}}}} \right\|}}} \right.
 \kern-\nulldelimiterspace} {\left\| {{{\bf{h}}_{\rm{B}}}} \right\|}}} \right\|}^2}}}{{{{\left\| {{{\bf{h}}_{\rm{B}}}} \right\|}^2}}}$ should be further satisfied to make $C_{\rm{S}}^ + \left| {_{\alpha=1}} \right.$ zero.

  Otherwise, in order  to make the log function  monotonically decrease,  we determine that  the condition (39) should be satisfied so that the maximum is achieved at $\alpha =0$. Obviously, $\alpha =0$ means that there is no transmission process and naturally indicates  that the available secrecy rate is equal to zero.

 Finally,  if log function in $0<\alpha<1$  is a constant with the variable $\alpha$, we determine  the condition (40) to make $C_{\rm{S}}^ - $ zero.  The proof is completed.
  \end{proof}
 Moreover, we define the available variable pair  as $\left( {\sigma _{\rm{E}}^2,\overline {\bf{U}}} \right)$ and collect those pairs that satisfy the above proposition to construct an achievable variable region (AVR) for Eve. Therefore, we have the following theorem.
  \begin{theorem}
  In noisy case, there always exists an AVR such that  MASR is  equal to zero.
   \end{theorem}
   It is a natural result from the Proposition 2 since the $C_{\rm{S}}^ - $ can be reduced to be zero when $0<\alpha<1$  and otherwise the $C_{\rm{S}}^ +$ can be controlled  to be zero by Eve through a suitable configuration of  ${\sigma _{\rm{E}}^2}$.
\section{Simulation Results}
\label{Simulation Results}
\begin{table}[t]
\caption{\label{tab:test}Simulation Parameters and Values}
\vspace{-15pt}
\begin{center}
\footnotesize
\begin{tabular*}{9cm}{@{\extracolsep{\fill}}ll}
  \toprule
 \multicolumn{1}{c}{ Simulation Parameters}  & \multicolumn{1}{c}{ Values }\\
  \midrule
    \multicolumn{1}{c}{Path-loss model}&\multicolumn{1}{c}{Eq. (26)}
  \\
     \multicolumn{1}{c}{Small-scale fading model}&\multicolumn{1}{c}{${\cal C}{\cal N}\left( {0,{{\bf{I}}_{{M}}}} \right)$}
  \\
      \multicolumn{1}{c}{Length $N$ of data }&\multicolumn{1}{c}{200 (symbols)}
  \\
  \multicolumn{1}{c}{Noiseless case: $d_{\rm B}$, $d_{\rm E}$}&\multicolumn{1}{c}{$\left[ {0,{{10}^4}} \right]$ (m)}
  \\
   \multicolumn{1}{c}{Noisy case: available region of $x$, $y$, $z$}&\multicolumn{1}{c}{$\left({ -\infty ,0} \right]$, $\left[ {0, \infty } \right)$, $\left[ {0,4} \right]$}
  \\
    \multicolumn{1}{c}{Maximum transmit power $P$}&\multicolumn{1}{c}{$35\left( {{\rm{dBm}}} \right)$}
  \\
  \multicolumn{1}{c}{Noisy case: noise power ${\widetilde \sigma _{\rm{B}}^2}$ (10MHZ bandwidth)}&\multicolumn{1}{c}{$ - 102\left( {{\rm{dBm}}} \right)$}
  \\
  \bottomrule
  \label{Table}
  \vspace{-15pt}
 \end{tabular*}
\end{center}
\end{table}
 \begin{figure*}[ht]
\begin{equation}\setcounter{equation}{43}
{{\bf{H}}_{\rm{E}}} = \left[ {\begin{array}{*{20}{c}}
{{\rm{0}}{\rm{.60 }} - {\rm{ 0}}{\rm{.80i}}}&{{\rm{0}}{\rm{.06  +  0}}{\rm{.61i}}}&{{\rm{1}}{\rm{.05 }} - {\rm{ 0}}{\rm{.28i}}}&{ - {\rm{0}}{\rm{.49  +  0}}{\rm{.16i}}}\\
{ - {\rm{0}}{\rm{.36  +  0}}{\rm{.47i}}}&{ - {\rm{0}}{\rm{.42  +  0}}{\rm{.29i}}}&{{\rm{0}}{\rm{.09 +  1}}{\rm{.53i}}}&{{\rm{0}}{\rm{.10 }} - {\rm{ 0}}{\rm{.59i}}}
\end{array}} \right]
\end{equation}
\vspace{-15pt}
 \end{figure*}
  \begin{figure*}[ht]
\begin{equation}
{{\bf{H}}_{\rm{L}}} = \left[ {\begin{array}{*{20}{c}}
{ - {\rm{0}}{\rm{.71  +  0}}{\rm{.12i}}}&{{\rm{0}}{\rm{.41  +  0}}{\rm{.94i}}}&{{\rm{0}}{\rm{.20  +  0}}{\rm{.91i}}}&{ - {\rm{1}}{\rm{.44  +  1}}{\rm{.49i}}}
\end{array}} \right],{\overline {\bf{U}}^{*}}  = \left[ {\begin{array}{*{20}{c}}
{ - \sqrt {{\rm{0}}{\rm{.2257}}} }&{ - \sqrt {{\rm{0}}{\rm{.7743}}} }\\
{ - \sqrt {{\rm{0}}{\rm{.7743}}} }&{\sqrt {{\rm{0}}{\rm{.2257}}} }
\end{array}} \right]
\end{equation}
\vspace{-15pt}
 \end{figure*}
In this section, we evaluate  the secrecy performance of $M$-1-2 wiretap channel model in the context of AN-based transmission mode at Alice. We consider the average MASR performance under  ergodic scenario  and also  determine a operation  region achieving zero secrecy rate under one  channel block.  Generally, the path loss  expressed in decibels   satisfies the following:
\begin{equation}\setcounter{equation}{41}
PL\left( d \right) = PL\left( {{d_0}} \right) + 10n{\log _{10}}\left( {\frac{d}{{{d_0}}}} \right)
\end{equation}
where $d_0$ is the reference point at 1 km and $n$ is known as the path loss exponent. In this paper, we consider a practical 3GPP urban model where $PL\left( {{d_0}} \right) = 148.1$ and $n = 3.76$.  ${\beta _{\rm{B}}}$  and ${\beta _{\rm{E}}}$ can be derived from ${d _{\rm{B}}}$ and  ${d _{\rm{E}}}$ based on equation (41). To better characterize the  operation region of noisy Eve, we define  several metrics as follows:
 \begin{equation}
x \buildrel \Delta \over = {\log _{10}}\left( {{{\widetilde \sigma _{\rm{E}}^2} \mathord{\left/
 {\vphantom {{\widetilde \sigma _{\rm{E}}^2} {\widetilde \sigma _{\rm{B}}^2}}} \right.
 \kern-\nulldelimiterspace} {\widetilde \sigma _{\rm{B}}^2}}} \right),y \buildrel \Delta \over = {\log _{10}}\left( {{{{\beta _{\rm{E}}}} \mathord{\left/
 {\vphantom {{{\beta _{\rm{E}}}} {{\beta _{\rm{B}}}}}} \right.
 \kern-\nulldelimiterspace} {{\beta _{\rm{B}}}}}} \right),z \buildrel \Delta \over = {\log _{10}}\left( {{d_{\rm{B}}}} \right)
\end{equation}
When Bob is located in  different regions $z$,  $x$ and $y$ represent the  measures that  noisy Eve could take  to reduce  the secrecy rate to be zero. The detailed simulation parameters can be shown in Table~\ref{Table}.
\begin{figure}[!t]
\centering \includegraphics[width=1.0\linewidth]{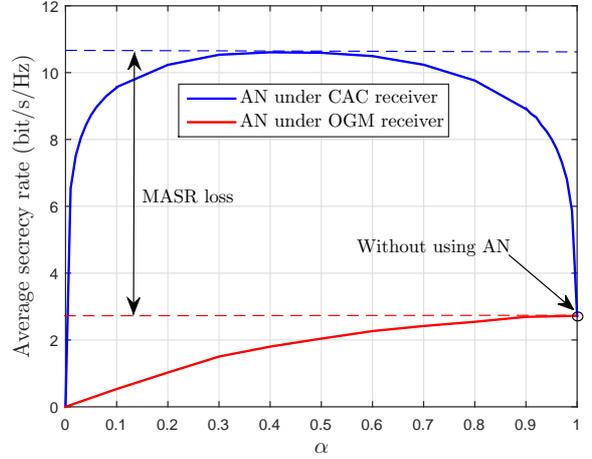}
\caption{Average  secrecy rate  versus  the power ratio of $\alpha$ under two types of receivers in the noisy environment.}
\label{Simulation0}
\vspace{-10pt}
\end{figure}
\begin{figure}[!t]
\centering\includegraphics[width=1.0\linewidth]{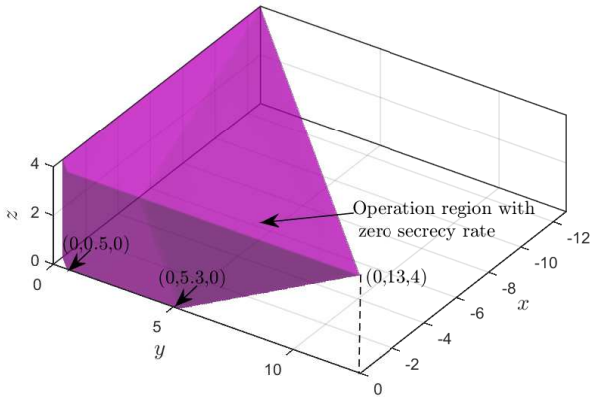}
\caption{ Operation region $\left( {x,y,z} \right)$  of noisy Eve under the channel model in equation (43) and (44) with the transform matrix ${\overline {\bf{U}}^{*}}$. Our strategy is to achieve the condition (37) and (38). }
\label{Simulation2}
\vspace{-10pt}
\end{figure}

Fig.~\ref{Simulation0} shows the average  secrecy rate  versus  the power ratio of $\alpha$ under different receivers at Eve. In this example, the number of transmit antennas satisfies  $M= 10$ and  the equivalent noise level of Bob is assumed to be equal to that of Eve, that's to say, $\sigma _{\rm{B}}^2=\sigma _{\rm{E}}^2$. We also  fix the SNR ${P \mathord{\left/
 {\vphantom {P { \sigma _{\rm{B}}^2}}} \right.
 \kern-\nulldelimiterspace} { \sigma _{\rm{B}}^2}}$ at 30 in $\rm{dB}$. As we can see, the average  secrecy rate under OGM filter increases  with $\alpha$ and achieves the maximum at $\alpha=1$ whereas  the MASR   under traditional CAC receiver is achieved  around at $\alpha=0.5$.  Basically, the MASR under  OGM receiver  is equal to that without using AN at Alice, which means that the AN scheme is completely broken down  by Eve.   This phenomenon in high SNR is in accord with the result shown in the equation (36).  The MASR loss can be as far as 8 bit/s/Hz.

 In Fig.~\ref{Simulation2}, we present the operation region $\left( {x,y,z} \right)$ which is available  for noisy Eve to cause zero secrecy rate.  Alice adopts the maximum transmission power shown in Table~\ref{Table}. For clarity, we consider one-time channel realization and  focus on  arbitrary  channel matrices. As an example, the simulated channel matrices and the generated  transform matrix are  shown in equation (43) and (44). Note that the transform matrix is always available and can be optimized based on the channel matrices and precoders at Eve. As we can see, the operation region is extensive. Especially, OGM based Eve can achieve zero secrecy rate when $z \in \left[ {0,4} \right]$ or equivalently ${d_{\rm{B}}} \in \left[ {1,{{10}^4}} \right]$, by flexibly controlling its path loss   and reducing the  noise level  so that  there exist  $y \in \left[ {0.5,5.3} \right]$ and  $x \in \left[ { - 13,0} \right]$.

\begin{figure}[!t]
\centering \includegraphics[width=1.0\linewidth]{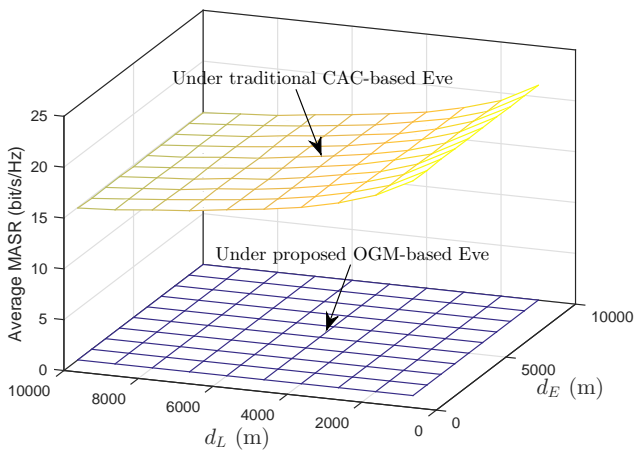}
\caption{Average MASR  versus the distance $d_{\rm B}$ and $d_{\rm E}$ in the noiseless case: $\sigma _{\rm{B}}^2=\sigma _{\rm{E}}^2 \to 0$.}
\label{Simulation1}
\vspace{-10pt}
\end{figure}

Fig.~\ref{Simulation1} demonstrates the  average MASR performance versus the distance measure $d_{\rm B}$ and $d_{\rm E}$.   Alice adopts the maximum transmission power shown in Table~\ref{Table}. The power allocation ratio in traditional AN based scheme is configured to be 0.5~\cite{Zhou} and the number of antennas is set to be $M=10$. In this case we assume that Bob has same noise level of Eve  and the noise level can approach zero. We can see that  the secrecy rate, even under  the scenario where Bob tries its best to reduce noise,  has to be  zero and becomes irrelevant to the distances when  the OGM-based Eve occurs. Therefore, the traditional expression of  secrecy rate  under AN-based scheme comes to be  imprecise and  the CAC receiver assumption for Eve actually causes a huge security for the system.
 \begin{figure}[!t]
\centering\includegraphics[width=1.0\linewidth]{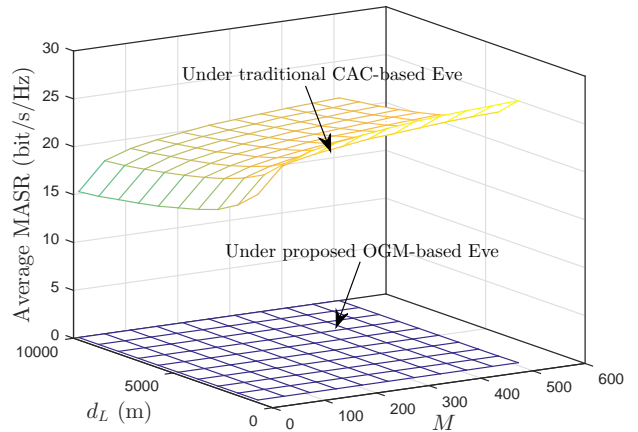}
\caption{ Average MASR  versus the distance $d_{\rm B}$ and $M$ in the noiseless case: $\sigma _{\rm{B}}^2=\sigma _{\rm{E}}^2 \to 0$.}
\label{Simulation3}
\vspace{-10pt}
\end{figure}

Fig.~\ref{Simulation3} shows the average MASR performance under various  $d_{\rm B}$ and $M$  in the noiseless Eve.  Alice adopts the maximum transmission power shown in Table~\ref{Table}. In terms of traditional AN based simulations, the location  of Eve is fixed at $d_{\rm E}=1000$ and the power ratio is fixed at $\alpha=0.55$ which is verified to be  optimal in~\cite{Zhou}. As we can see, the MASR under traditional AN based scheme is gradually enhanced with the increase  of $d_{\rm B}$ and $M$. However, the MASR  under OGM-based Eve is always equal to be zero. This result is irrelevant to the variation of $d_{\rm B}$ and the increase of $M$. This is because the high-SINR level induced by the OGM filter can suppress   the high array gains which are  brought by Alice for Bob and usually guaranteed by very large number of transmit antennas.
\section{Conclusions}
\label{Conclusions}
In this paper, we introduced a novel eavesdropping behavior which is targeted for AN based secure transmission.  We demonstrated how Eve utilized  a series of optimal filters to combat the influence of AN and proved  that the maximum achievable secrecy rate which is subjected to  this advanced eavesdropper was compelled to be zero in the case of both  noiseless  and noisy case. Simulation results  showed that  both the  distance measure, either from  Alice  to Bob or to Eve,  actually could not  eliminate or reduce this disaster impact. This terrible phenomenon also occurred  even though  the number of transmit antennas is increased significantly.
%



\begin{thebibliography}{99}
\bibitem{Wyner_AD} A. D. Wyner, ``The wire-tap channel," \emph{Bell Sys. Tech. J.}, vol. 54, pp. 1355-1387, Oct. 1975.
\bibitem{Leung-Yan-Cheong} S. Leung-Yan-Cheong and M. Hellman, ``The Gaussian wire-tap channel,"  \emph{IEEE Trans. Inf. Theory}, vol. 24, no. 4, pp. 451-456, July, 1978.
\bibitem{Yingbin_Liang}Y. Liang, H. V. Poor, and  S. Shamai, `` Secure communications over fading channels," \emph{IEEE Trans. Inf. Theory}, vol. 54, no. 6, pp. 2470-2492, Jun. 2008.
\bibitem{Amitav_Mukherjee} A. Mukherjee, S. A. Fakoorian, J. Huang, and A. L. Swindlehurst, ¡°Principles of physical layer security in multiuser wireless networks: A survey," \emph{IEEE Commun. Surveys Tuts.}, vol. 16, no. 3, pp.1550-1573, Mar. 2014.
\bibitem{Khisti_I}A. Khisti and G. Wornell, ``Secure transmission with multiple antennas I: The MISOME wiretap channel," \emph{IEEE Trans. Inf. Theory}, vol. 56, no. 7, pp. 3088-3104, Jul. 2010.
\bibitem{Negi}R. Negi and S. Goel, ``Guaranteeing secrecy using artificial noise," \emph{IEEE Trans. Wireless Commun.}, vol. 7, no. 6, pp. 2180-2189, Jun. 2008.
\bibitem{Zhou}X. Zhou and M. R. McKay, ``Secure transmission with artificial noise over fading channels: Achievable rate and optimal power allocation," \emph{IEEE Trans. Veh. Technol.}, vol. 59, no. 8, pp. 3831-3842, Oct. 2010.
\bibitem{Khisti_II}A. Khisti and G. W. Wornell,  ``Secure transmission with multiple antennas¡ªPart II: The MIMOME wireltap channel," \emph{IEEE Trans. Inf. Theory}, vol. 56, no. 11, pp. 5515-5532, Nov. 2010.
\bibitem{Nan_Yang} N. Yang, M. Elkashlan, T. Q. Duong, J. Yuan, and R. Malaney, ``Optimal transmission with artificial noise in MISOME wiretap channels," \emph{ IEEE Trans. Veh. Technol.}, vol. 65, no. 4, pp. 2170-2181, Apr. 2016.
\bibitem{Jun_Zhu}  J. Zhu, R. Schober, and V. K. Bhargava, ``Secure transmission in multicell massive MIMO systems,"  \emph{IEEE Trans. Wireless Commun.}, vol. 13, no. 9, pp. 4766-4781, Sep. 2014.
\end{thebibliography}
\end{document}